\documentclass[11pt]{article}
\usepackage{amsmath, amssymb, amsthm}
\usepackage{graphicx}
\usepackage[hidelinks]{hyperref}
\usepackage{natbib}
\usepackage{algorithm}
\usepackage{algpseudocode}
\usepackage{xcolor}
\usepackage{listings}
\usepackage{graphicx}
\usepackage{float}

\graphicspath{{figures/}}

\DeclareMathOperator*{\argmax}{arg\,max}

\newtheorem{definition}{Definition}
\newtheorem{theorem}{Theorem}
\newtheorem{corollary}{Corollary}

\title{Procela: Epistemic Governance in Mechanistic Simulations Under Structural Uncertainty}
\author{Kinson Vernet \\ \texttt{kinson.vernet@gmail.com} \\Independent Researcher}
\date{March 27, 2026}

\newcommand{\preprintnotice}{This is a preprint of an article currently under review at \textit{Advanced Theory and Simulations}. The final published version will be available at [future DOI].}

\begin{document}

\maketitle

\preprintnotice

\begin{abstract}
Mechanistic simulations typically assume fixed ontologies: variables, causal relationships, 
and resolution policies are static. This assumption fails when the true causal structure is 
contested or unidentifiable--as in antimicrobial resistance (AMR) spread, where contact, 
environmental, and selection ontologies compete. We introduce Procela, a Python framework 
where variables act as epistemic authorities that maintain complete hypothesis memory, 
mechanisms encode competing ontologies as causal units, and governance invariants observe 
epistemic signals and mutate system topology at runtime. This is the first framework where 
simulations test their own assumptions. We instantiate Procela for AMR in a hospital network 
with three competing families. Governance detects coverage decay, policy fragility, 
and run structural probes. Results show 20.4\% error reduction and 69\% cumulative regret 
improvement over baseline. All experiments are reproducible with full auditability. Procela 
establishes a new paradigm: simulations that model not only the world but their own 
modeling process, enabling adaptation under structural uncertainty.
\end{abstract}

\section{Introduction}

\subsection{The problem with fixed-ontology simulation}

Mechanistic simulation has become indispensable for understanding complex systems, 
from infectious disease spread to climate dynamics and economic forecasting. 
The standard paradigm assumes fixed ontologies: variables, causal relationships, 
and resolution mechanisms are specified before execution and remain static 
throughout \citep{epstein2008why, north2013complex}. 

This assumption fails when the true causal structure of the system is contested, 
unidentifiable, or subject to change. As \citet{box1976science} famously observed, 
``all models are wrong"—but when multiple causal explanations compete, the 
challenge is not merely approximation but fundamental uncertainty about which 
model is even appropriate. In antimicrobial resistance (AMR) spread, 
for example, multiple explanations compete: transmission via patient contact, 
environmental reservoirs, or antibiotic selection pressure 
\citep{lipsitch2002antimicrobial, austin1999relationship}. Each implies different 
interventions, yet observational data may not distinguish them \citep{pearl2009causality}. 
Forcing a choice among these ontologies discards epistemic uncertainty at the 
moment of model specification.

\subsection{Current approaches and their limits}

Existing methods for handling model uncertainty operate within the fixed-ontology paradigm:

\begin{itemize}
    \item \textbf{Model averaging} \citep{hoeting1999bayesian} combines predictions 
    from multiple fixed models but does not question the models themselves.
    \item \textbf{Robust control} \citep{hanse2019robust} optimizes for worst-case 
    outcomes across a fixed set of scenarios.
    \item \textbf{Adaptive management} \citep{walters1986adaptive} treats interventions 
    as experiments but leaves the underlying model structure unchanged.
    \item \textbf{Bayesian structural learning} \citep{heckerman1995learning} discovers 
    causal graphs from data but requires the true structure to be identifiable.
\end{itemize}

What is missing is an architecture where the simulation itself can question its own 
assumptions—where the model is not a fixed artifact but an evolving hypothesis about 
the world, capable of testing competing causal explanations against observational 
evidence \citep{pearl2009causality, box1976science}.

\subsection{Procela: A new paradigm}

We introduce Procela, a Python 3.10+ simulation architecture that elevates 
epistemic uncertainty from a pre-simulation concern to a runtime process. Procela is 
built on four core abstractions:

\begin{itemize}
    \item \textbf{Variables as memory-bearing epistemic authorities}. Unlike traditional variables 
    that merely hold values, Procela variables maintain a complete history of competing 
    hypotheses, each with associated confidence and source. They resolve conflicts via 
    defined policies (weighted voting, highest confidence, etc.), making arbitration 
    explicit and auditable.

    \item \textbf{Mechanisms as scientific theories}. Each mechanism encodes a causal 
    hypothesis: it reads variables, applies a transformation, and writes hypotheses to 
    variables. Multiple mechanisms can compete to explain the same phenomenon, encoding 
    different scientific perspectives.

    \item \textbf{Governance as first-class citizen}. Invariants and hooks observe the system's 
    epistemic state--disagreement among mechanisms, coverage decay, prediction conflicts, 
    etc.--and can trigger actions on invariant violations: adding or removing, enabling or 
    disabling mechanisms, switching resolution policies, adding or removing invariants, etc. 
    This turns simulation into an ongoing scientific process.

    \item \textbf{Executive}. It orchestrates the entire simulation: executes mechanisms, 
    resolves variables, evaluates invariants at PRE/RUNTIME/POST phases, calls pre/post hooks. 
    It maintains random state, step counter, and logging system history.
\end{itemize}

\subsection{Contributions}

This paper makes four contributions:

\begin{enumerate}
    \item \textbf{Conceptual architecture}. We formalize the notion of epistemic governance 
    in simulation, defining variables as memory-bearing epistemic authorities that arbitrate 
    competing hypotheses and governance as system mutation based on epistemic signals.

    \item \textbf{Procela implementation}. We present an open-source Python framework 
    implementing these abstractions, with full auditability and runtime governance.

    \item \textbf{AMR case study}. We instantiate Procela for antimicrobial resistance 
    spread in a hospital networks, with three competing ontology families (contact, 
    environmental, selection) and three levels of governance (policy fragility, coverage decay, 
    structural probe).

    \item \textbf{Empirical demonstration}. We show that governance reduces prediction 
    error by 20.4\% with experiments succeeding 54--71\% of the time, intervening effectively 
    during high errors and reverting when experiments fail.
\end{enumerate}

\subsection{Paper roadmap}

Section 2 presents Procela's core abstractions formally. Section 3 describes governance 
as a first-class capability. Section 4 introduces the AMR case study design. 
Section 5 defines the epistemic signals used in our AMR instantiation. Section 6 details 
the three levels of governance used in our AMR study. Section 7 reports results. Section 8 
discusses implications and limitations. Section 9 concludes. Appendices provide mathematical 
foundations, and implementation details.

\subsection{Reproducibility}

Procela is open-source under the Apache 2.0 license at \url{https://procela.org}. 
All code for the AMR case study, including governance and visualization scripts, 
is available at \url{https://procela.org/amr}. A browser-based demonstration allowing 
readers to run experiments interactively is available at \url{https://procela.org/runner}.

\section{Procela core abstractions}\label{sec:procela_core}

Procela is built on four fundamental abstractions: variables as memory-bearing epistemic 
authorities, mechanisms as causal units, governance via invariants and hooks as first-class 
citizen allowing dynamic mutation of the system during execution based on epistemic
signals, and executive that orchestrates simulation 
execution while maintaining full auditability. We present each formally.

\subsection{Variables as memory-bearing epistemic authorities}

In traditional simulation, a variable is a container for a value. In Procela, a variable 
is an \textit{epistemic authority} that \textit{memorizes} the full history of competing 
hypotheses and resolves them according to a policy.

\begin{definition}[Variable]
A variable $V$ is a tuple $(\mathcal{D}, \mathcal{M}, \pi, k)$ where:
\begin{itemize}
    \item $\mathcal{D}$ is the domain of possible values. Domains may be continuous 
    intervals $[a,b] \subset \mathbb{R}$ or discrete sets $\{v_1,\ldots,v_n\}$.
    \item $\mathcal{M}$ is a memory: an ordered sequence of hypotheses, resolved 
    conclusion and reasoning result $\{{(h_{i,1}, h_{i,2}, \ldots, h_{i,t}), c_i, r_i}\}$ 
    where each $i$ tuple corresponds to simulation step $i$.
    \item $\pi$ is a resolution policy: a function $\pi: \mathcal{M}(t) \rightarrow \mathcal{D}$ 
    that arbitrates among the competing hypotheses and gives a single resolved 
    conclusion at step $t$.
    \item $k$ is a unique cryptographic key identifying this variable
\end{itemize}
\end{definition}

Each hypothesis $h \in \mathcal{M}(t)$ is itself a structured object.

\begin{definition}[Hypothesis]
A hypothesis $h$ is an immutable tuple $(v, c, s, t, k)$ where:
\begin{itemize}
    \item $v \in \mathcal{D}$ is the proposed value
    \item $c \in [0,1]$ is the confidence associated with this hypothesis
    \item $s$ is a source identifier (typically the cryptographic key of 
    the mechanism that generated it)
    \item $t$ is the timestamp (simulation step) when the hypothesis was generated
    \item $k$ is a unique cryptographic key identifying this hypothesis
\end{itemize}
\end{definition}

The memory $\mathcal{M}$ grows monotonically; hypotheses are never deleted, ensuring 
full auditability. At any step, one can query $\mathcal{M}(t)$ to retrieve all hypotheses 
generated for that step, or $\mathcal{M}(t_1, t_2)$ for a range.

\subsubsection{Resolution policies}

A resolution policy $\pi$ arbitrates among competing hypotheses. Common policies include:

\begin{itemize}
    \item \textbf{Highest confidence}: $\pi(\mathcal{H}) = \argmax_{h \in \mathcal{H}} c_h$, 
    breaking ties arbitrarily.
    \item \textbf{Confidence voting}: $\pi(\mathcal{H}) = \frac{\sum_{h \in \mathcal{H}} c_h \cdot v_h}{\sum_{h \in \mathcal{H}} c_h}$ (for continuous domains).
    \item \textbf{Median}: $\pi(\mathcal{H}) = \text{median}(\{v_h : h \in \mathcal{H}\})$.
    \item \textbf{Source priority}: $\pi$ selects based on a predefined hierarchy of sources.
\end{itemize}

Policies can be changed at runtime via governance (Section~\ref{sec:governance}), enabling 
dynamic adaptation of resolution strategies.

\subsection{Mechanisms as causal units}

A mechanism encodes a causal hypothesis: given current variable states, it proposes new 
values for variables.

\begin{definition}[Mechanism]
A mechanism $M$ is a tuple $(R, W, f, k, \text{active})$ where:
\begin{itemize}
    \item $R \subseteq \mathbb{V}$ is a set of variables read
    \item $W \subseteq \mathbb{V}$ is a set of variables written
    \item $f: \prod_{r \in R} \mathcal{D}_r \times \Theta \rightarrow \prod_{w \in W} \mathcal{H}_w$ is 
    a transformation function, where $\Theta$ represents any internal parameters 
    (including random state)
    \item $k$ is a unique cryptographic key identifying this mechanism
    \item $\text{active} \in \{\text{True}, \text{False}\}$ indicates whether the mechanism 
    is currently enabled
\end{itemize}
\end{definition}

At each simulation step, every active mechanism is executed. The function $f$ reads the 
variables in $R$, applies a transformation (which may be deterministic, stochastic, 
or stateful), and generates one hypothesis for each variable in $W$. Each hypothesis is tagged 
with the mechanism's key $k$ and the current confidence (which may be a function of 
internal state).

Mechanisms may be stateful, maintaining internal parameters that evolve over time. 
This allows, for example, adaptive mechanisms that learn from prediction errors.

\subsubsection{Auditability and source tracking}

Every hypothesis in every variable's memory is permanently recorded with its source key. 
This creates a complete audit trail:

\begin{itemize}
    \item Which mechanism proposed what value, when, and with what confidence
    \item How each variable resolved competing hypotheses at each step
    \item Which governance actions occurred and when
\end{itemize}

Source keys are cryptographic, enabling verification that a hypothesis genuinely came 
from a claimed mechanism. This is essential for trust in distributed or collaborative 
simulation settings.

The full history $\mathcal{H}$ is queryable. One can, for example, retrieve all hypotheses 
generated by mechanism $M$ for variable $V$ between steps $t_1$ and $t_2$, or compute 
the accuracy of each mechanism's predictions over time.

\subsection{Governance as a first-class citizen}
\label{sec:governance}

Traditional simulation frameworks treat the model as fixed once specified.
Procela inverts this: the model itself is a dynamic entity that can be observed,
interrogated, and mutated during execution. This capability is \textit{epistemic governance}:
the modification of system (structure, policy, etc.) during simulation based on epistemic signals.

Procela provides two types of governance:
\begin{enumerate}
    \item invariants: second-class objects that observe system state and trigger actions ;
    \item hook events: first-class objects that are designed to trigger structural changes.
\end{enumerate}

\subsubsection{Invariants}

\begin{definition}[Invariant]
An invariant $I$ is a tuple $(\phi, \alpha, \rho, \sigma)$ where:
\begin{itemize}
    \item $\phi: \mathcal{H} \rightarrow \{\text{True}, \text{False}\}$ 
    is a condition function that evaluates the current history $\mathcal{H}$
    \item $\alpha: \mathcal{H} \rightarrow \mathcal{IV}$ is a 
    function that triggers an action (raise violation errors, switch policy, etc.).
    $\mathcal{IV}$ is the set of invariant violations in Procela.
    \item $\rho \in \{\text{PRE}, \text{RUNTIME}, \text{POST}\}$ is the phase when the 
    invariant is evaluated
    \item $\sigma \subseteq \mathbb{V} \cup \mathbb{I}$ is the scope of 
    variables and invariants this invariant can observe
\end{itemize}
\end{definition}

The condition $\phi$ returns False when the invariant is \textit{violated}--that is, 
when the system is in a state that requires attention. When violation occurs, 
the action $\alpha$ is executed.

Unlike hooks, invariants are restricted to operating on variable views rather than directly 
manipulating the system structure. They are designed to access the current snapshot, providing 
sufficient information to detect trends, accelerations, and regime changes within a timestep. 
While invariants can influence the system's structure, they do so under restricted conditions, 
requiring users to carefully adhere to the governance rules they define.

\textbf{Phases of evaluation}

The three phases correspond to different moments in the simulation step 
(Algorithm~\ref{alg:executive}):

\begin{itemize}
    \item \textbf{PRE}: Evaluated before any mechanisms execute. Used for invariants 
    that must prevent certain states from occurring (e.g., ``never allow more than 100 
    colonized patients").

    \item \textbf{RUNTIME}: Evaluated after mechanisms generate hypotheses but before 
    variables resolve. Used for invariants that react to emerging patterns in 
    hypotheses (e.g., ``if mechanisms disagree strongly, switch resolution policy").

    \item \textbf{POST}: Evaluated after variables resolve and the step completes. 
    Used for invariants that analyze completed steps, compute metrics, or prepare for 
    future governance (e.g., check mass conservation or update coverage statistics).
\end{itemize}

Multiple invariants may coexist within the same phase. In the current implementation, 
their execution order is determined by the sequence in which they are added.

\subsubsection{Hooks}

Procela provides two hook events (pre\_step, post\_step) to govern a system at 
the highest level of governance. They are the first-class objects designed to observe 
the system state and can trigger structural changes. 

\begin{definition}[Hooks]
A hook is a function $f: S \rightarrow  S'$ that evaluates the current system 
state $\mathcal{S}$ and produces a new state $S'$ of the system.
\end{definition}

Crucially, hooks have access to the complete system structure, not just the current snapshot. 
This enables complete mutation of the system before and after each timestep.

\subsubsection{Mutation operations}

A governance can perform any of the following mutations on the system state:

\begin{itemize}
    \item Mechanism: enable, disable, add, remove, or update parameters
    \item Variable: set a policy $\pi$, modify a variable domain (rare, but possible)
    \item Invariant: add, remove, or update internal parameters
    \item Executive: set random number generator, reset the world state (this represents the 
    creation of a new world instance over the same structure, use with caution but 
    useful for Monte Carlo study)
\end{itemize}

All mutations are recorded in the system history, creating an audit trail of governance actions.

\subsubsection{The governance lifecycle}

A typical governance intervention follows a structured pattern:

\begin{enumerate}
    \item \textbf{Detection}: A condition indicates that the system requires attention.
    \item \textbf{Action}: An action is executed or the system state is mutated.
    \item \textbf{Observation}: The system continues executing under the new state.
    \item \textbf{Evaluation}: The system assesses whether the governance improved outcomes.
    \item \textbf{Conclusion}: Based on evaluation, the mutation may be kept, reverted, or 
    modified.
\end{enumerate}

This lifecycle mirrors the scientific method: observe a phenomenon, formulate a hypothesis, 
run an experiment, evaluate evidence, update beliefs.

\subsubsection{Epistemic observability}

Notably, Procela defines no built-in epistemic metrics. The architecture provides the 
\textit{mechanisms} for governance (invariants, hooks, mutations, phases) but not the 
\textit{content} of what to observe. Epistemic signals such as coverage, fragility, 
and conflict emerge from user-defined variables and mechanisms that compute them.

This design choice ensures that Procela remains domain-agnostic. What counts as an 
epistemic signal in epidemiology (coverage of transmission models) may differ from 
climate science (ensemble spread) or economics (forecast variance). Users define what 
to observe; Procela provides the tools to observe, decide, and act. Futur versions
may introduce some domain metrics.

\subsection{Executive}

The executive orchestrates the entire simulation.

\begin{definition}[Executive]
The executive $\mathcal{E}$ maintains:
\begin{itemize}
    \item $\mathbb{V}$: the set of all variables in the system
    \item $\mathbb{M}$: the set of all mechanisms, with their active/inactive status
    \item $\mathbb{I}$: the set of all governances
    \item $t$: the current simulation step (non-negative integer)
    \item $rng$: a random number generator state for reproducibility
    \item $logger$: a logging module for audit in different formats (JSON, file or console).
\end{itemize}
\end{definition}

The executive operates in discrete steps. The simulation step is abstract, allowing 
discrete and continuous time simulation. Algorithm~\ref{alg:executive} shows the execution loop.

\begin{algorithm}
\caption{Executive simulation step}\label{alg:executive}
\begin{algorithmic}[1]
\Require Executive $\mathcal{E}$ at step $t$
\State Snapshot current values of all $V \in \mathbb{V}$ (for read consistency)
\State Evaluate all invariants in phase PRE
\For{each active mechanism $M \in \mathbb{M}$}
    \State Read current values of $R_M$
    \State Compute $f_M$ to generate hypotheses for $W_M$
    \State Append hypotheses to memories of each $w \in W_M$
\EndFor
\State Evaluate all invariants in phase RUNTIME
\For{each variable $V \in \mathbb{V}$}
    \State Let $\mathcal{H}_V$ be hypotheses for $V$ at step $t$
    \State Compute resolved value $v^* = \pi_V(\mathcal{H}_V)$
    \State Append $v^*$ to $V$'s memory as resolved value
\EndFor
\State $t \leftarrow t + 1$
\State Evaluate all invariants in phase POST
\end{algorithmic}
\end{algorithm}

The three invariant phases and the hooks allow governance before, during, 
and after state updates. This enables governance that prevent undesirable states, 
react to emerging patterns or perform cleanup and analysis after the step completes.

\subsection{Discussion}

These four abstractions--variables as memory-bearing epistemic authorities, mechanisms as 
causal units, governance as first-class citizen, and executive that orchestrates the 
simulation--form a minimal but powerful foundation. The architecture imposes no particular 
scientific theory; it can be seen as a meta-framework within which competing theories can 
coexist, compete, and be governed.

Procela defines no built-in epistemic metrics (coverage, fragility, etc.). These 
emerge from the instantiation: users define variables to track such quantities and mechanisms 
to compute them. This design keeps the core framework domain-agnostic while enabling rich 
epistemic monitoring in applications.

Epistemic governance as a first-class citizen transforms simulation from static ontologies 
into an adaptive, self-questioning process. Procela enables:

\begin{itemize}
    \item \textbf{Structural learning}: The system can test whether disabling certain 
    mechanisms improves outcomes.
    \item \textbf{Epistemic risk management}: Governance can intervene when uncertainty 
    threatens decision quality.
    \item \textbf{Auditable science}: Every governance action is recorded, creating a 
    complete history of model evolution.
    \item \textbf{Transferable patterns}: Governance strategies (e.g., ``probe by temporary 
    disablement") can be reused across domains.
\end{itemize}

The next section instantiates these abstractions for the AMR case study, showing how 
domain-specific signals and epistemic governance emerge from the core framework.

\section{Case study: antimicrobial resistance spread in hospital networks}
\label{sec:case_study}

We instantiate Procela for a canonical problem in infectious disease epidemiology: the spread 
of antimicrobial-resistant organisms in a hospital network. This domain is characterized by 
fundamental uncertainty about transmission mechanisms, making it an ideal testbed for 
epistemic governance.

\subsection{Domain context}

Antimicrobial resistance (AMR) is a global health crisis. In hospital settings, resistant 
organisms spread through multiple potential pathways:

\begin{itemize}
    \item \textbf{Patient-to-patient contact}: Direct transmission via healthcare workers, 
    shared equipment, or patient proximity.
    \item \textbf{Environmental reservoir}: Contamination of surfaces, water, or air that 
    persists and colonizes new patients.
    \item \textbf{Antibiotic selection pressure}: Antimicrobial use kills susceptible strains, 
    allowing resistant strains to proliferate.
\end{itemize}

Each pathway implies different interventions: isolation for contact transmission, enhanced 
cleaning for environmental reservoirs, and antibiotic stewardship for selection pressure. 
Yet observational data often cannot distinguish which pathway dominates at a given 
time \citep{lipsitch2002antimicrobial}. A simulation that must choose a single ontology 
discards this uncertainty.

\subsection{System scope}

We model a simplified hospital network with three interconnected units ($H_1, H_2, H_3$). 
Patients transfer between units, antibiotics are administered, and environmental contamination 
accumulates and decays. The time horizon is 160 discrete steps, sufficient to observe multiple 
regime shifts.

\subsection{Variables}

We define five core variables, each an instance of a Procela variable as defined in 
Section~\ref{sec:procela_core}. All are observable measurements, not interpretations--they 
represent what can be directly measured in a hospital.

\begin{itemize}
    \item \textbf{Colonized patients} $C \in [0, 100]$: Number of patients colonized with a 
    resistant organism. This is the primary outcome of interest.
    
    \item \textbf{Antibiotic usage} $A \in [0, 50]$: Defined daily doses (DDD) of antibiotics 
    administered per time step. Represents selection pressure.
    
    \item \textbf{Environmental load} $E \in [0, 100]$: Composite index of environmental 
    contamination, aggregating surface samples, air counts, and water tests.
    
    \item \textbf{Intervention code} $I \in \{0, 1, 2, 3\}$: Current infection control policy:
    \begin{itemize}
        \item $0$: Baseline (standard precautions)
        \item $1$: Isolation (contact precautions for colonized patients)
        \item $2$: Enhanced cleaning (environmental decontamination)
        \item $3$: Antibiotic stewardship (restrict usage)
    \end{itemize}
    
    \item \textbf{Predicted colonization} $\hat{C} \in [0, 100]$: One-step-ahead prediction 
    of colonized patients.
\end{itemize}

The first three variables are set by a hidden ground truth (section \ref{sec:ground-truth}). 
Competing ontologies propose hypotheses to the last two variables which use the 
\texttt{WeightedConfidencePolicy} as default resolution, unless changed by governance.

\subsection{Mechanism families}

Three mechanism families encode competing ontologies. Each family consists of multiple 
mechanism variants (e.g., noisy, biased, lagged) that share the same core equation but 
differ in parameters. This intra-family diversity creates realistic variation in predictions 
and recommendations.

\subsubsection{Contact family}

The contact family embodies the hypothesis that transmission is primarily driven by patient 
contacts.

Core equation:
\begin{equation}
\hat{C}_{t+1} = C_t + \beta_C C_t (1 - \eta_C \mathbf{1}_{I_t=1}) + \epsilon_C
\label{eq:contact}
\end{equation}

where:
\begin{itemize}
    \item $\beta_C > 0$ is the contact transmission rate
    \item $\eta_C \in [0,1]$ is the effectiveness of isolation (intervention code 1)
    \item $\epsilon_C \sim \mathcal{N}(0, \sigma_C)$ is process noise
    \item $\mathbf{1}_{I_t=1}$ is an indicator for isolation being active
\end{itemize}

Policy recommendation: isolation ($I=1$) with confidence proportional to the mechanism's 
current belief in its own accuracy.

\subsubsection{Environmental family}

The environmental family posits that environmental contamination drives colonization.

Core equation:
\begin{equation}
\hat{C}_{t+1} = C_t + \beta_E E_t (1 - \eta_E \mathbf{1}_{I_t=2}) + \epsilon_E
\label{eq:environmental}
\end{equation}

where:
\begin{itemize}
    \item $\beta_E > 0$ scales environmental contribution
    \item $\eta_E \in [0,1]$ is cleaning effectiveness (intervention code 2)
    \item $\epsilon_E \sim \mathcal{N}(0, \sigma_E)$
\end{itemize}

Policy recommendation: enhanced cleaning ($I=2$).

\subsubsection{Selection family}

The selection family attributes spread to antibiotic pressure selecting for resistant strains.

Core equation:
\begin{equation}
\hat{C}_{t+1} = C_t + \beta_S A_t (1 - \eta_S \mathbf{1}_{I_t=3}) + \epsilon_S
\label{eq:selection}
\end{equation}

where:
\begin{itemize}
    \item $\beta_S > 0$ scales antibiotic contribution
    \item $\eta_S \in [0,1]$ is stewardship effectiveness (intervention code 3)
    \item $\epsilon_S \sim \mathcal{N}(0, \sigma_S)$
\end{itemize}

Policy recommendation: antibiotic stewardship ($I=3$).

The figure \ref{fig:mech_contrib} shows the mechanism contributions (governance combined) 
for the AMR case study.

\begin{figure}[H]
\centering
\includegraphics[width=0.8\textwidth]{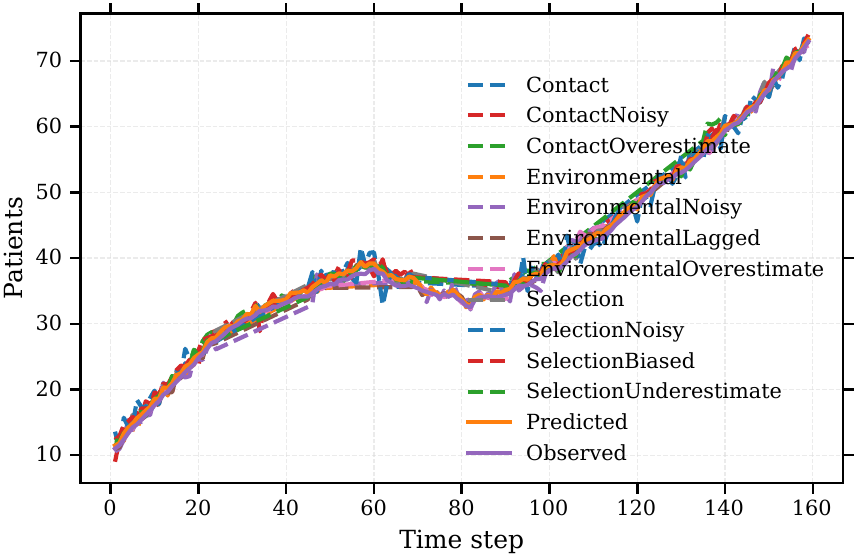}
\caption{Mechanism contributions (governance combined) for the AMR case study. The predicted 
colonized patients is for combined governances (Section \ref{sec:amr_governance}).}
\label{fig:mech_contrib}
\end{figure}

\subsection{Ground truth regime shifts}\label{sec:ground-truth}

The simulated world (unknown to mechanisms) undergoes three regime shifts:

\begin{itemize}
    \item \textbf{Steps 0--60: Selection regime}. Antibiotic pressure dominates. $\beta_S$ 
    is high; $\beta_C$ and $\beta_E$ are low.
    
    \item \textbf{Steps 61--110: Environmental regime}. An environmental contamination event 
    occurs (e.g., plumbing failure). $\beta_E$ becomes high; $\beta_S$ and $\beta_C$ return 
    to baseline.
    
    \item \textbf{Steps 111--160: Contact regime}. A patient transfer outbreak begins. $\beta_C$ 
    becomes high; $\beta_E$ returns to baseline.
\end{itemize}

\begin{figure}[H]
\centering
\includegraphics[width=0.8\textwidth]{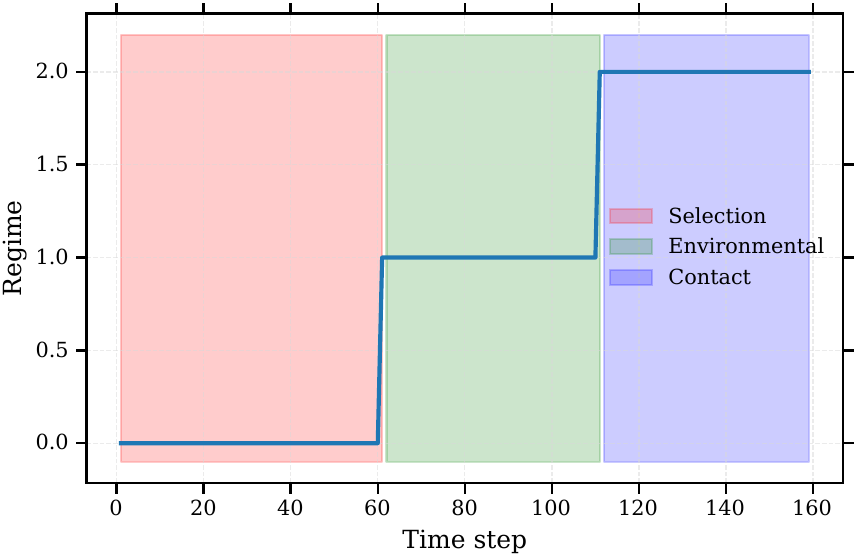}
\caption{Regime shifts for the AMR case study.}
\label{fig:regime_shifts}
\end{figure}

These shifts are implemented by varying the true $\beta$ parameters in the data-generating 
process as illustrated in figure \ref{fig:regime_shifts}. Mechanisms do not observe these 
parameters; they must infer the current regime from data.

\subsection{Resolution policies}

Two resolution policies are used in this case study:

\begin{itemize}
    \item \textbf{WeightedConfidencePolicy} (default): $\pi(\mathcal{H}) = \frac{\sum_{h \in \mathcal{H}} c_h v_h}{\sum_{h \in \mathcal{H}} c_h}$. 
    The resolved value is a confidence-weighted average of all hypotheses at each timestep.
    
    \item \textbf{HighestConfidencePolicy} (used in experiments): $\pi(\mathcal{H}) = v_{h^*}$ 
    where $h^* = \argmax_{h \in \mathcal{H}} c_h$. The hypothesis with highest confidence 
    determines the value.
\end{itemize}

The choice of policy affects how competing ontologies are arbitrated. Weighted voting averages 
across perspectives; highest confidence selects a single winner.

\subsection{Discussion}

This design instantiates Procela's core abstractions for the AMR domain:

\begin{itemize}
    \item The variables are Procela variables with domains and policies.
    \item Three mechanism families encode competing scientific theories as Procela mechanisms.
    \item The data-generating process (regime shifts) is external to the simulation--mechanisms 
    only observe its outputs.
\end{itemize}

Notably, we have not yet defined epistemic signals (coverage, fragility) or governance. 
These emerge in the next two sections.

The design is minimal but sufficient: five variables, three mechanism families, two 
policies, and a non-stationary ground truth. This simplicity ensures that any observed 
benefits of governance can be attributed to the framework, not domain complexity.

\section{Epistemic signals for AMR governance}
\label{sec:epistemic_signals}

Governance requires observability. To decide when and how to intervene, invariants and hooks 
must detect when the system's epistemic state is deteriorating--when mechanisms are failing, 
disagreeing, or becoming unreliable. This section defines the epistemic signals we use in 
the AMR case study and shows how they are implemented in Procela.

\subsection{The need for epistemic observability}

In the AMR domain, the ground truth (which transmission regime is active) is hidden. Governance 
cannot observe regime shifts directly. It must infer them from patterns in the data and in 
mechanism behavior. Epistemic signals are computable quantities that correlate with regime 
shifts and indicate when governance intervention may be beneficial.

Crucially, these signals are not built into Procela's core. They emerge from user-defined 
variables and mechanisms that observe the system's behavior. This design ensures that what 
counts as an ``epistemic signal" is domain-specific and customizable.

\subsection{Coverage: prediction accuracy per mechanism}

Coverage measures how accurately a mechanism predicts the observed value of $C$. For a single 
mechanism $m$ at step $t$, we define:

\begin{equation}
\text{cov}(m, t) = \exp\left(-\frac{|\hat{C}_{m,t-1} - C_t|}{\sigma_{\text{cov}}}\right)
\label{eq:coverage}
\end{equation}

where:
\begin{itemize}
    \item $\hat{C}_{m,t-1}$ is the hypothesis generated by mechanism $m$ at step $t-1$ for 
    variable $\hat{C}$
    \item $C_t$ is the observed value of colonized patients at step $t$
    \item $\sigma_{\text{cov}} > 0$ is a scaling parameter that determines how quickly 
    coverage decays with error
\end{itemize}

Coverage ranges from near 0 (large prediction error) to 1 (perfect prediction). The exponential 
transformation ensures that coverage decays smoothly as error increases.

For a mechanism family $F$ (e.g., all contact mechanisms), we define family-level coverage as 
the average across its members:

\begin{equation}
\text{cov}(F, t) = \frac{1}{|F|} \sum_{m \in F} \text{cov}(m, t)
\label{eq:family_coverage}
\end{equation}

Coverage decay--a sustained decrease in $\text{cov}(F, t)$--suggests that the ontology encoded 
by family $F$ no longer matches the current regime.

\subsection{Policy fragility: disagreement on intervention}

Policy fragility measures how much mechanisms disagree about the appropriate intervention $I$. 
At step $t$, let $\mathcal{H}_I(t)$ be the set of hypotheses for variable $I$. Each hypothesis 
proposes an intervention code $v \in \{0,1,2,3\}$ with confidence $c$.

We define fragility as the normalized range of proposed interventions:

\begin{equation}
\text{frag}(t) = \frac{\max(\{v_h : h \in \mathcal{H}_I(t)\}) - \min(\{v_h : h \in \mathcal{H}_I(t)\})}{3}
\label{eq:fragility}
\end{equation}

The denominator 3 normalizes by the maximum possible range (from 0 to 3). Fragility ranges 
from 0 (all mechanisms recommend the same intervention) to 1 (mechanisms span the full range 
from no intervention to stewardship).

High fragility indicates that the system's epistemic state is not consolidated: different 
ontologies imply different actions. This is a precursor to policy oscillation and 
decision regret.

\subsection{Recent error: current prediction accuracy}

Recent error is a simple moving average of absolute prediction error:

\begin{equation}
\text{err}_{\text{recent}}(t) = \frac{1}{w} \sum_{i=t-w+1}^{t} |\hat{C}_i - C_i|
\label{eq:recent_error}
\end{equation}

where $w$ is a window size. This signal captures the current performance of the resolved 
prediction $\hat{C}$, aggregating across all active mechanisms.

Recent error serves two purposes:
\begin{itemize}
    \item As a **trigger condition**: Governance should intervene when recent error is high, 
    indicating systemic failure.
    \item As an **evaluation metric**: After an experiment, comparing recent error before 
    and after determines whether the intervention succeeded.
\end{itemize}

\subsection{Implementing epistemic signals}

These signals are not magic--they are implemented as ordinary Procela variables.

\subsubsection{Coverage variables}

For each mechanism family, we create a variable to track its coverage:

\begin{verbatim}
family.coverage = Variable(
    name="family_name_coverage",
    domain=RangeDomain(0.0, 1.0),
    policy=None  # observational, not resolved
)
\end{verbatim}

\subsubsection{Policy fragility}

Similarly, a policy fragility variable tracks the disagreement signal:

\begin{verbatim}
policy_fragility = Variable(
    name="policy_fragility",
    domain=RangeDomain(0.0, 1.0),
    policy=None
)
\end{verbatim}

\subsubsection{Prediction error on colonization}

\begin{verbatim}
error_colonized = Variable(
    name="error_colonized",
    domain=RealDomain(),
    policy=None
)
\end{verbatim}

This variable holds the prediction error on colonization after each step, allowing full 
auditability. During governance, a rolling window of $|\hat{C} - C|$ is used from this 
variable memory.

These epistemic variables are updated in the POST phase, after variables are resolved conflicts.
They serve as governance inputs. Because these are ordinary Procela variables, they maintain 
full memories. Invariants and hooks can detect trends (e.g., ``coverage has declined 
for 10 consecutive steps") not just instantaneous values.

\subsection{Discussion}

This design illustrates a key principle of Procela: **the framework provides the tools for 
epistemic observability**. Users define what signals matter in their domain and implement 
observers to compute them. This ensures:

\begin{itemize}
    \item **Domain relevance**: Signals are tailored to the specific scientific questions.
    \item **Extensibility**: New signals can be added without modifying the framework.
    \item **Auditability**: Every signal's computation is traceable.
\end{itemize}

In the next section, we show how governance uses these signals to detect epistemic 
crises and runs structural experiments.

\section{AMR governance}
\label{sec:amr_governance}

We implement three levels of governance for the AMR case study, each encoding a different 
hypothesis about how to improve system performance. All governances follow the scientific 
method pattern: detect a potential crisis, formulate a hypothesis, run an experiment, 
evaluate evidence, and conclude.

\subsection{Common patterns}

In the AMR study, all governance levels share several design elements:

\begin{itemize}
    \item \textbf{State machine}: Each maintains internal state (\texttt{monitoring}, 
    \texttt{experimenting}, \texttt{evaluating}) to track experiment lifecycle.
    
    \item \textbf{Temporary experiments}: Mutations are applied for a fixed duration, then 
    evaluated. Failed experiments are reverted from the next timestep.
    
    \item \textbf{Evidence-based decisions}: Experiments succeed only if they reduce mean 
    absolute prediction error.
    
    \item \textbf{Audit trail}: All experiments are logged with pre/post error, hypothesis, 
    and outcome.
\end{itemize}

\subsection{Policy fragility governance}

The policy fragility governance addresses situations where mechanisms disagree strongly about 
which intervention to apply.

\subsubsection{Detection}

The governance triggers when:
\begin{itemize}
    \item Policy fragility exceeds threshold $\tau_{\text{frag}}$: $\text{frag}(t) > \tau_{\text{frag}}$
    \item Recent error exceeds threshold $\tau_{\text{err}}$: $\text{err}_{\text{recent}}(t) > \tau_{\text{err}}$
\end{itemize}

The fragility threshold is set to $\tau_{\text{frag}} = 0.6$, which corresponds to mechanisms 
spanning almost the full range of interventions (from 0 to 3, normalized). The error threshold is 
$\tau_{\text{err}} = 1.0$, indicating moderate prediction error.

\subsubsection{Hypothesis}

\begin{quote}
``Switching from weighted confidence to highest confidence resolution will reduce prediction error 
during periods of high disagreement.''
\end{quote}

\subsubsection{Experiment}

\begin{enumerate}
    \item Record current policy of variable $I$ (default: \texttt{WeightedConfidencePolicy}).
    \item Switch $I$'s policy to \texttt{HighestConfidencePolicy}.
    \item Run for $d_{\text{exp}} = 10$ steps.
    \item Collect prediction errors during this period.
\end{enumerate}

\subsubsection{Evaluation}

After the experiment period, compute:
\begin{itemize}
    \item $\text{err}_{\text{pre}}$: mean absolute error in the $d_{\text{eval}} = 10$ steps 
    before the experiment
    \item $\text{err}_{\text{exp}}$: mean absolute error during the experiment
\end{itemize}

If $\text{err}_{\text{exp}} < \text{err}_{\text{pre}}$, the experiment succeeds and the new 
policy is retained. Otherwise, the original policy is restored.

\subsection{Coverage decay}

The coverage decay detects when an entire ontology family is consistently failing 
to predict accurately.

\subsubsection{Detection}

For each family $F \in \{\text{contact}, \text{environmental}, \text{selection}\}$, this 
governance monitors coverage $\text{cov}(F, t)$. A family is considered ``decaying'' if:

\begin{equation}
\text{cov}(F, t) < \tau_{\text{cov}} \quad \text{for at least } k \text{ consecutive steps}
\label{eq:decay}
\end{equation}

We use $\tau_{\text{cov}} = 0.85$ and $k = 3$ in our experiments.

\subsubsection{Hypothesis}

\begin{quote}
``Family $F$ is no longer valid in the current regime. Temporarily disabling it will improve 
overall predictions.''
\end{quote}

\subsubsection{Experiment}

\begin{enumerate}
    \item Identify the decaying family $F$.
    \item Disable all mechanisms in $F$.
    \item Run for $d_{\text{exp}} = 10$ steps.
    \item Monitor prediction error during this period.
\end{enumerate}

\subsubsection{Evaluation}

Compare mean absolute error in the $d_{\text{eval}} = 10$ steps before disablement versus 
during disablement. If error decreases, the family remains disabled; otherwise, it is re-enabled.

\subsection{Structural probe}

The structural probe takes a more systematic approach: it periodically isolates 
each family to measure its standalone predictive performance.

\subsubsection{Detection}

Unlike the previous governance, this one does not wait for a crisis. It runs on a fixed 
schedule: every $p = 25$ steps, it initiates a probe of one family, cycling through families 
in round-robin order.

\subsubsection{Hypothesis}

\begin{quote}
``Measuring a family's isolated performance will reveal if its ontology currently best 
matches the regime, informing future governance decisions.''
\end{quote}

\subsubsection{Experiment}

For the chosen family $F$:
\begin{enumerate}
    \item Record current system state (all families active).
    \item Disable all families except $F$.
    \item Run for $d_{\text{probe}} = 20$ steps.
    \item Record mean absolute error during this isolated period.
    \item Restore original topology for the next schedule.
\end{enumerate}

\subsubsection{Evaluation}

The probe does not permanently change the system; it merely collects data. The resulting 
isolated error for each family can be stored in a dedicated variable:

\begin{equation}
\text{iso\_err}(F, t) = \text{MAE during isolation of } F
\label{eq:iso_error}
\end{equation}

This information can be used by other invariants or by human analysts to understand which 
ontology is performing best in the current regime.

\subsection{Governance composition}

These three governance levels can operate simultaneously. The policy fragility and coverage decay 
are reactive, triggering only when epistemic signals indicate a problem. The 
structural probe is proactive, running on a schedule regardless of system state.

Their potential conflicts can be managed through:
\begin{itemize}
    \item \textbf{Phase separation}: Structural probe runs in PRE, policy fragility runs in RUNTIME, and coverage decay runs in POST phase.
    \item \textbf{State awareness}: Governance checks whether an experiment is already 
    active before triggering.
\end{itemize}

In the AMR case study, we run each governance separately to isolate their effects, as reported in 
Section~\ref{sec:results}. We also show how adding all the three governances would react.
If conflicts disable all mechanisms, an emergency governance is triggered to enable 
all of them.

\subsection{Discussion}

These three levels of governance illustrate the range of governance strategies possible in Procela:

\begin{itemize}
    \item \textbf{Policy fragility} changes how conflicts are resolved.
    \item \textbf{Coverage decay} removes hypothesized causes.
    \item \textbf{Structural probe} gathers evidence about competing theories.
\end{itemize}

All follow the same scientific pattern: detect, hypothesize, experiment, evaluate, conclude. 
All are implemented using only the Procela abstractions--variables, mechanisms, policies, 
invariants, hooks--plus the epistemic signals defined in Section~\ref{sec:epistemic_signals}.

The next section reports the results of running these governances in the AMR case study.

\section{Results}
\label{sec:results}

We evaluate the three governances on the AMR case study described in 
section~\ref{sec:case_study}. Each simulation runs for 160 steps with three regime 
shifts (selection, environmental, contact). Results are averaged over 50 independent 
runs with different random seeds. All code and data are available at 
\url{https://procela.org/amr}.

\subsection{Baseline: no governance}

As a baseline, we run the simulation with all three mechanism families active and the 
default \texttt{WeightedConfidencePolicy} for all variables. No governance is active. 
This represents the traditional approach: multiple models compete, but the simulation 
structure is static. Table~\ref{tab:baseline} shows prediction error.

\begin{table}[H]
\centering
\begin{tabular}{|l|c|c|}
\hline
\textbf{Metric} & \textbf{Value} & \textbf{95\% CI} \\
\hline
Mean absolute error & 0.535 & [0.516, 0.554] \\
Standard deviation & 0.383 & [0.370, 0.396] \\
\hline
\end{tabular}
\caption{Baseline prediction error with no governance.}
\label{tab:baseline}
\end{table}

All the family ontologies show lower prediction error during selection regime than during 
environmental and contact regimes. Contact family shows higher errors in environmental and 
contact regimes as depicted in figure \ref{fig:coverage_none}.

\begin{figure}[H]
\centering
\includegraphics[width=\textwidth]{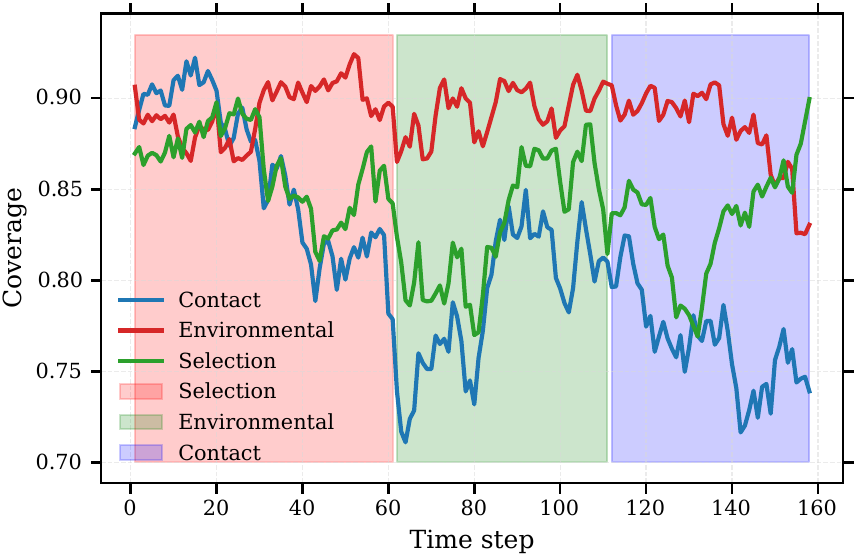}
\caption{Coverage timeline for the baseline.}
\label{fig:coverage_none}
\end{figure}

\subsection{Policy fragility governance}

The policy fragility governance (Section~\ref{sec:amr_governance}) triggers when intervention 
disagreement and recent prediction error exceed thresholds ($\tau_{\text{frag}}=0.6$, $\tau_{\text{err}}=1.0$). 
It experiments by switching $I$'s policy to \texttt{HighestConfidencePolicy} for 10 steps, 
then evaluates. The table \ref{tab:fragility} presents the performance of fragility metrics 
compared to the baseline.

\begin{table}[H]
\centering
\begin{tabular}{|l|c|c|}
\hline
\textbf{Metric} & \textbf{Value} & \textbf{Improvement w.r.t baseline} \\
\hline
Mean absolute error & 0.544 $\pm$ 0.019 & \textbf{-1.636\%} \\
Standard deviation & 0.386 $\pm$ 0.013 & -0.712\% \\
Experiments triggered & 7 $\pm$ 0 & -- \\
Successful experiments & 5 $\pm$ 0 & 71.4\% success rate \\
\hline
\end{tabular}
\caption{Policy fragility governance results in comparison to the baseline.}
\label{tab:fragility}
\end{table}

The policy fragility governance fails to improve prediction error because it intentionally 
substitutes a higher-resolution policy with a lower-resolution alternative. Under an i.i.d. 
assumption, the weighted-confidence resolution is optimal (Theorem \ref{theo:voting_optimal}). 
Therefore, switching to highest confidence policy results in degraded predictive accuracy.

\begin{figure}[H]
\centering
\includegraphics[width=\textwidth]{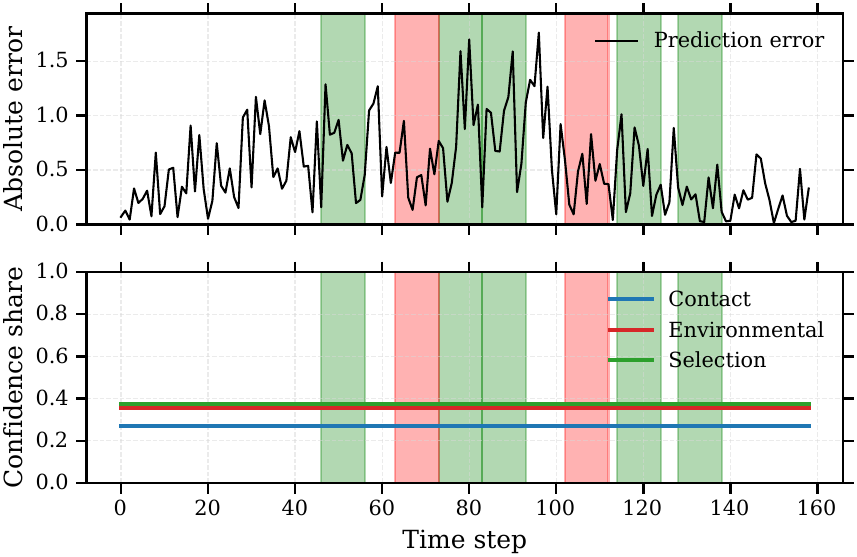}
\caption{Timeline of policy fragility experiments and confidence share between ontologies. 
Shaded regions indicate experiment periods for 1 experiment. Successes (green) occur in 
high-error regimes.}
\label{fig:fragility_timeline}
\end{figure}

The figure~\ref{fig:fragility_timeline} shows a representative run with five successful experiments 
(high error) and two failed experiments during stable periods. Despite the success rate (71.4\%) of 
experiments, this governance fails because when experiments fail, the prediction errors are much 
higher than when they succeed.

\subsection{Coverage decay governance}

The coverage decay monitors each family's prediction accuracy. When a family's 
coverage falls below $\tau_{\text{cov}}=0.85$ for $k=3$ consecutive steps, it temporarily 
disables that family for 10 steps.

\begin{table}[H]
\centering
\begin{tabular}{|l|c|c|}
\hline
\textbf{Metric} & \textbf{Value} & \textbf{Improvement w.r.t baseline} \\
\hline
Mean absolute error & 0.426 $\pm$ 0.036 & \textbf{20.42\%} \\
Standard deviation & 0.327 $\pm$ 0.027 & \textbf{14.76\%} \\
Experiments triggered & 11 $\pm$ 0 & -- \\
Successful experiments & 6 $\pm$ 0 & 54.5\% success rate \\
\hline
\end{tabular}
\caption{Coverage decay governance results.}
\label{tab:coverage}
\end{table}

The table \ref{tab:coverage} shows that coverage decay achieves lower mean error than 
the baseline with lower variance. This achievement for 54\% success rate means 
when experiments are successful, it works well; when it fails, no catastrophic failures.

By comparison to the results with no governance, coverage shows that contact family is 
continuously disabled from the selection regime. It means that after each 
experiment, this family's coverage remains below the given threshold allowing governance to 
keep it disabled as illustrated in figure \ref{fig:coverage_metrics}. As a consequence, 
disabling the contact family keeps the selection family above the threshold 
($\tau_{\text{cov}}=0.85$).

\begin{figure}[H]
\centering
\includegraphics[width=0.49\textwidth]{figures/figure3}
\includegraphics[width=0.49\textwidth]{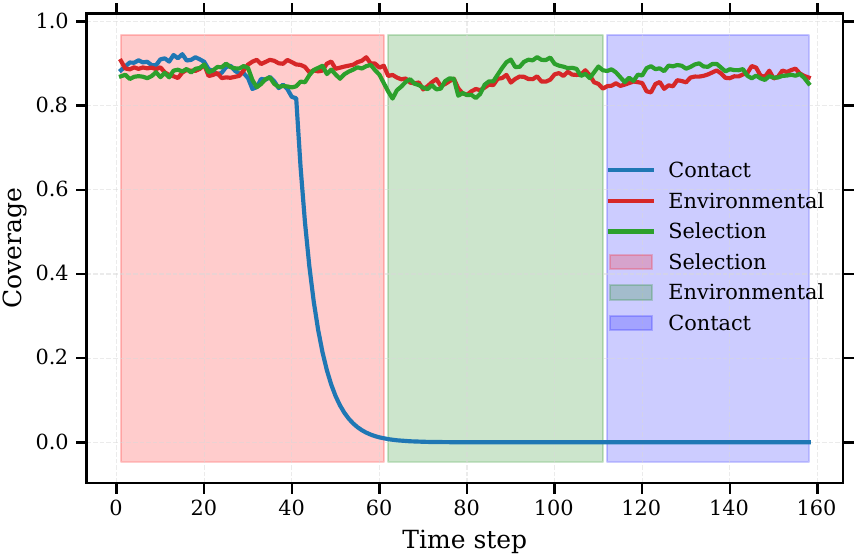}
\caption{Experiment outcomes for coverage governance. (Left): coverage for baseline 
(no governance), (Right): coverage for coverage decay governance. By comparison to the results 
with no governance, coverage shows that contact family is continuously disabled from the 
selection regime.}
\label{fig:coverage_metrics}
\end{figure}

The figure~\ref{fig:coverage_timeline} shows absolute error trajectories for the coverage decay 
governance with six successful experiments over 11 triggered. Successful experiments occur 
during high errors and fail during stable periods.

\begin{figure}[H]
\centering
\includegraphics[width=\textwidth]{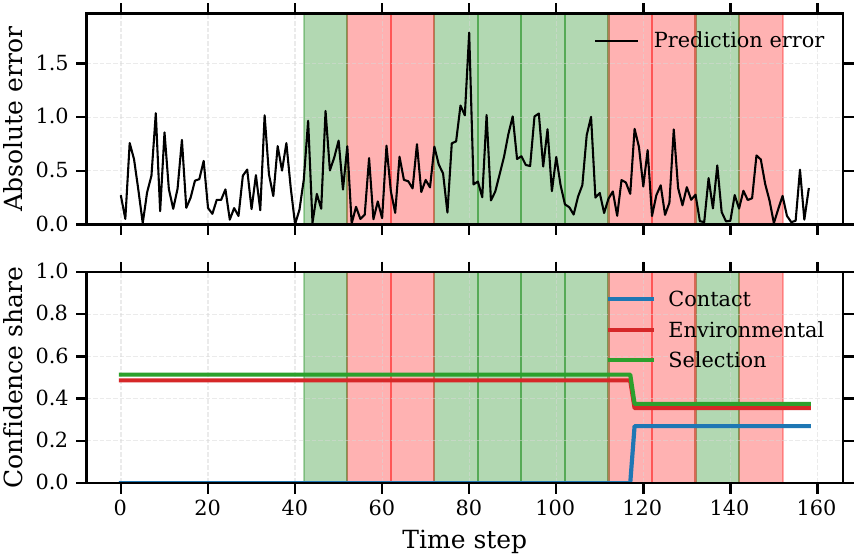}
\caption{Absolute error and confidence share over time for the three families under coverage 
 decay governance. Shaded regions indicate when a family was disabled. Coverage shows a total 
 of six successful experiments over 11 triggered during different regimes.}
\label{fig:coverage_timeline}
\end{figure}

\subsection{Structural probe governance}

The structural probe governance runs on a fixed schedule (every 25 steps), isolating each 
family in turn for 20 steps to measure its standalone performance.

\begin{table}[H]
\centering
\begin{tabular}{|l|c|c|}
\hline
\textbf{Metric} & \textbf{Value} & \textbf{Improvement w.r.t baseline} \\
\hline
Mean absolute error & 0.485 $\pm$ 0.020 & 9.31\% \\
Standard deviation & 0.364 $\pm$ 0.018 & 5.00\% \\
Probes triggered & 3 $\pm$ 0 & -- \\
Successful probes & 2 $\pm$ 0 & 66.67\% \\
\hline
\end{tabular}
\caption{Structural probe governance results in comparison to the baseline.}
\label{tab:probe}
\end{table}

The table \ref{tab:probe} shows that structural probe yields a modest reduction in error 
relative to coverage decay. However, its main contribution extends beyond predictive 
performance: it provides structural information. Separately, we later demonstrate that it 
also improves decision-making outcomes, as measured by cumulative difference 
(Section \ref{fig:cumulative_difference}).

\begin{figure}[H]
\centering
\includegraphics[width=\textwidth]{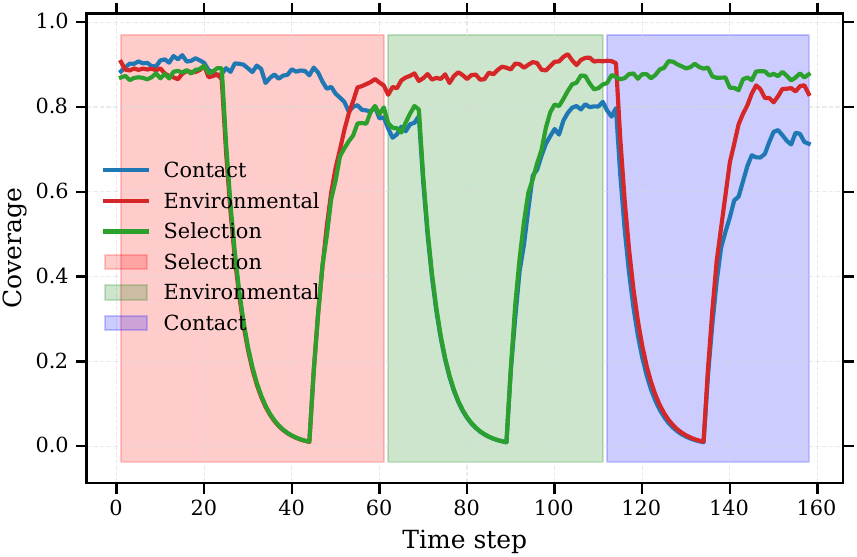}
\caption{Isolated performance for each family across probes. Isolation is chosen to avoid 
total coincidence to regime. During selection regime, contact family is isolated to estimate 
its performance. During environmental regime, environmental family is isolated. During contact 
regime, selection family is isolated.}
\label{fig:probe}
\end{figure}

Isolation is chosen to avoid total coincidence to regime as illustrated in figure~\ref{fig:probe}. 
During selection regime, contact family is isolated to estimate its performance. During 
environmental regime, environmental family is isolated. During contact regime, selection family 
is isolated. Figure~\ref{fig:probe_timeline} shows structural probe trajectories with 2 successful 
experiments over 3 triggered. Successful experiments occur in regimes where family errors 
are high (environmental and contact) and fail in regimes where families are stable (selection).

\begin{figure}[H]
\centering
\includegraphics[width=\textwidth]{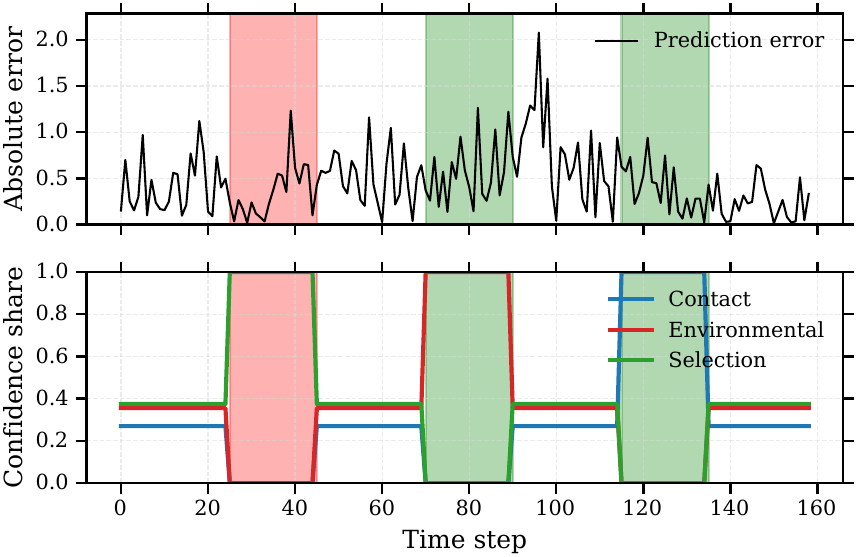}
\caption{Absolute error and confidence share for structural probe over time. Shaded regions 
indicate when a family was isolated. Probe shows two successful experiments during environmental 
and contact regimes (family errors are higher) and 1 failed experiment duration selection 
regime (family errors are lower).}
\label{fig:probe_timeline}
\end{figure}

\subsection{Combined governance}

Combining all the governances we define in the AMR study shows that the mean error is still 
improving (8.18\%) and the standard deviation is slightly below the baseline. Combining 
governances adds instability in comparison to the other governances isolated as depicted in 
table \ref{tab:governance_all}.

\begin{table}[H]
\centering
\begin{tabular}{|l|c|c|}
\hline
\textbf{Metric} & \textbf{Value} & \textbf{Improvement w.r.t baseline} \\
\hline
Mean absolute error & 0.491 $\pm$ 0.030 & 8.18\% \\
Standard deviation & 0.378 $\pm$ 0.034 & 1.31\% \\
\hline
\end{tabular}
\caption{Combined governances results in comparison to the baseline.}
\label{tab:governance_all}
\end{table}

The figure \ref{fig:governance_all_metrics} shows the combined governances in action.
In comparison to the results for structural probe, the coverage shows that a family
might be not completely isolated indicating that either policy fragility governance is enabling 
or coverage governance is still acting. Despite including an ineffective governance component, 
the combined governance tries to reconstruct the intervention code.

\begin{figure}[H]
\centering
\includegraphics[width=\textwidth]{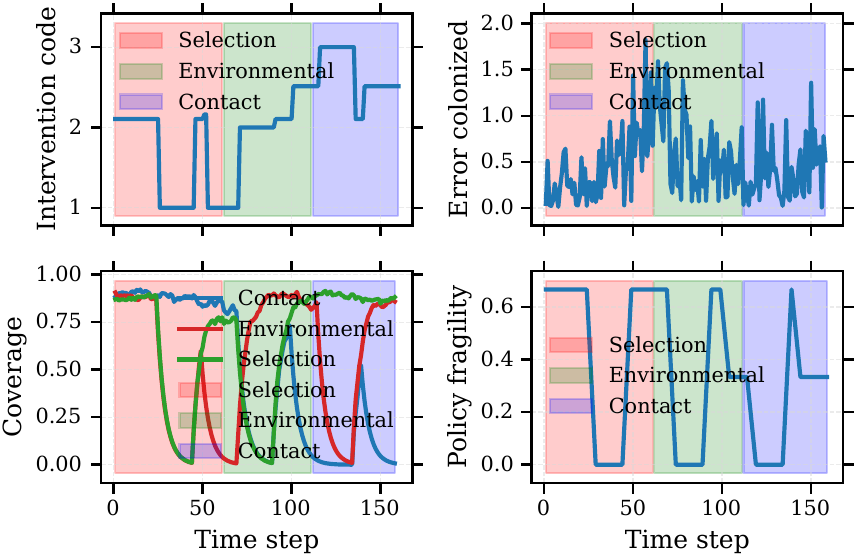}
\caption{Performance for the combined governances. (Top, left): intervention code over time 
shows how combined governances tries to reconstruct the intervention despite including an 
ineffective governance component; (Top, right): absolute prediction error over time; 
(Bottom, left): coverage indicates that a family might be not completely isolated; 
(Bottom, right): policy fragility over time. }
\label{fig:governance_all_metrics}
\end{figure}

The figure \ref{fig:governance_all_timeline} shows successful and failed experiments 
when combining all the governances in the AMR study. Almost all experiments failed during
the selection regime where family coverages are high. When family coverages are low 
(environmental and contact), experiments are likely to be succeeded. This result demonstrates 
that different governances can coexist in Procela to dynamically govern the system during 
the simulation.

\begin{figure}[H]
\centering
\includegraphics[width=\textwidth]{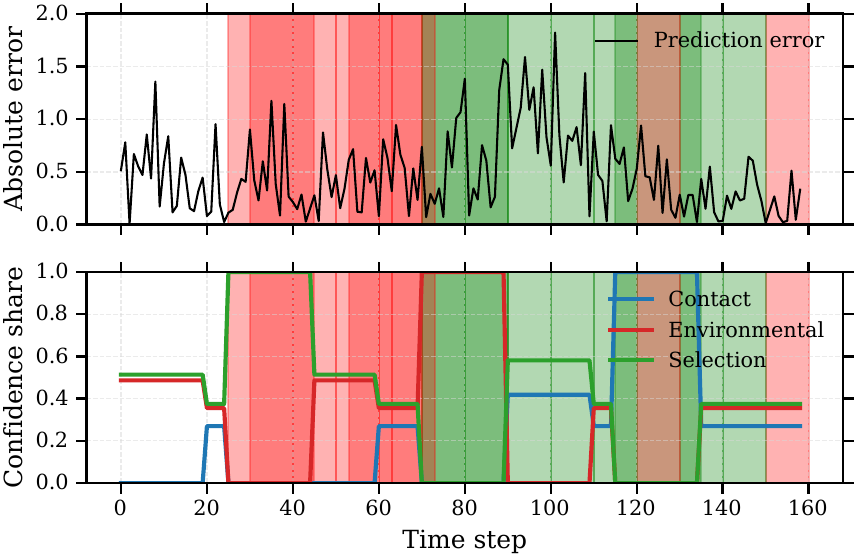}
\caption{Performance timeline of the combined governances. Almost all experiments failed during
the selection regime where family coverages are high. When family coverages are low 
(environmental and contact), experiments are likely to be succeeded.}
\label{fig:governance_all_timeline}
\end{figure}

\subsection{Cumulative difference}\label{sec:cumulative_difference}

We assess decision-making performance through cumulative difference relative to the optimal 
intervention (intervention code = 0). Formally, cumulative difference is defined as the 
aggregated difference between the realized outcomes under the proposed strategy and 
those obtained under the optimal intervention applied at each timestep. For convenience of 
visualization, we use the following formula for cumulative difference:

\begin{equation}
R(T) = \sum_{t=1}^T (e_t^{\text{optimal}} - e_t^{\text{actual}})
\end{equation}

The cumulative difference performance for each governance in comparison to the baseline 
(no governance) is shown in table \ref{tab:cumulative}. Structural probe yields the 
highest performance (69\%) while coverage decay yields only 35.5\%.

\begin{table}[H]
\centering
\begin{tabular}{|l|c|c|}
\hline
\textbf{Strategy} & \textbf{Mean error} & \textbf{Cumulative difference} \\
\hline
Policy fragility & -1.63\% & -9.8\% \\
Coverage decay & \textbf{20.42\%} & 35.5\% \\
Structural probe & 9.31\% & \textbf{69.0\%} \\
Combined governance & 8.18\% & 61.6\%\\
\hline
\end{tabular}
\caption{Decision-making performance through cumulative difference in comparison to the baseline.}
\label{tab:cumulative}
\end{table}

The cumulative difference over time is illustrated in figure \ref{fig:cumulative_difference}. 
The case study has revealed a fundamental trade-off between epistemic and decision-theoretic 
optimization:

\begin{itemize}
    \item Prediction-optimal $\neq$ decision-optimal
    \item Coverage decay makes better predictions but makes modest decisions (+20.42\%, C.R=+35.5\%)
    \item Structural probe makes better decisions but only modestly improves predictions 
    (+9.31\%, C.E=+69.0\%)
\end{itemize}

\begin{figure}[H]
\centering
\includegraphics[width=0.49\textwidth]{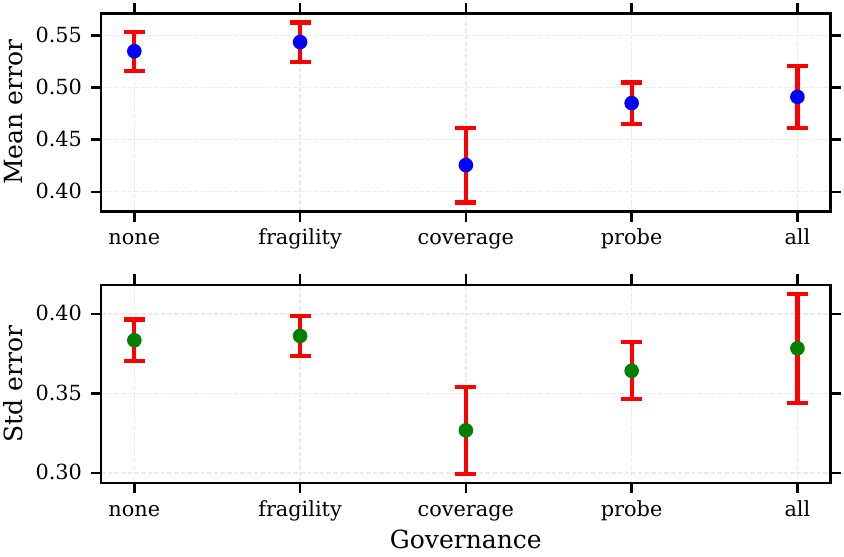}
\includegraphics[width=0.49\textwidth]{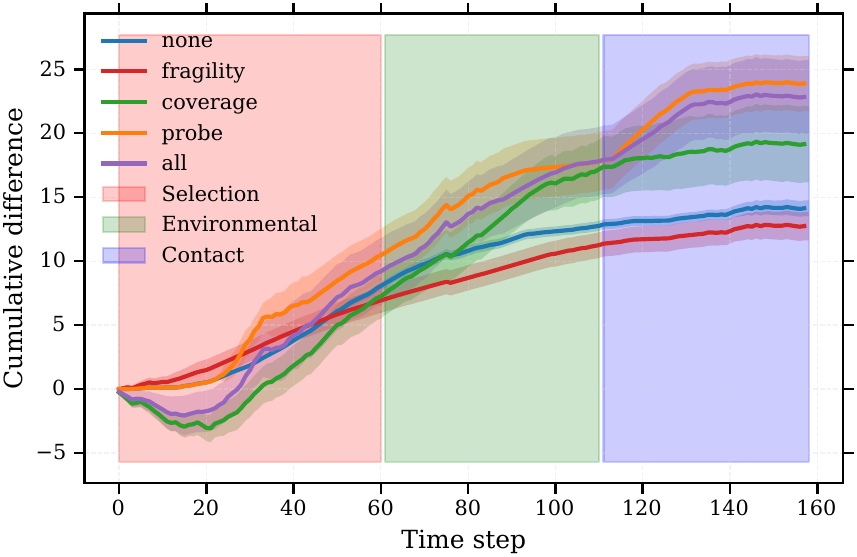}
\caption{(Left): Comparative performance of governance strategies for our AMR case study. 
(Right): Cumulative difference relative to the optimal intervention (intervention code = 0) 
over time. Coverage decay makes better predictions (+20.42\%) while structural probe makes 
better decisions (C.E=+69.0\%).}
\label{fig:cumulative_difference}
\end{figure}

Although the coverage decay governance reduces the average error, it initially makes poor 
intervention choices compared to optimal intervention. The cumulative difference metric 
penalizes these. Probe governance wins on cumulative difference because it almost always makes 
better decisions than the optimal intervention. This means:

\begin{itemize}
    \item Probe rarely makes bad intervention choices
    \item But its average error improvement is modest
    \item It's conservative and avoids disasters
\end{itemize}

This result suggests to consider hybrid strategy -- probe governance for information, coverage 
for prediction.

\subsection{Topology dynamics in the AMR study}

Procela ensures complete auditability. The memories of variables can be used to 
visualize the topology dynamics. As illustrated in figure \ref{fig:topology_evolution} 
the topology evolution for each governance in the AMR study is depicted. The contact family 
has three mechanisms, while the environmental and selection families have four each. In total, 
the AMR case study has eleven competing mechanisms that governance restructures over time.

\begin{figure}[H]
\centering
\includegraphics[width=\textwidth]{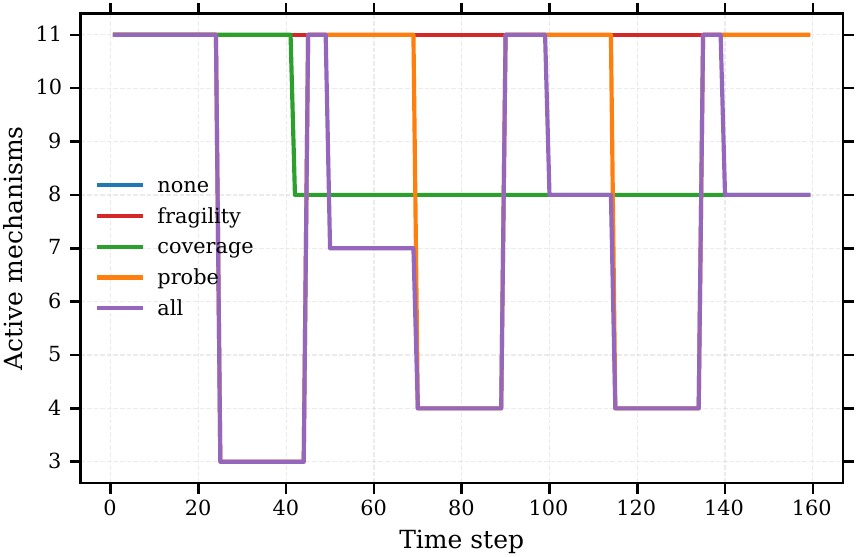}
\caption{Topology evolution for each governance in the AMR study. The contact family has three mechanisms, 
while the environmental and selection families have four each. In total, the AMR case study has 
eleven competing mechanisms.}
\label{fig:topology_evolution}
\end{figure}

\subsection{Comparative analysis}

A comparative performance of governance strategies for our AMR case study is shown in table 
\ref{tab:comparative}.

\begin{table}[H]
\centering
\begin{tabular}{|l|c|c|c|c|}
\hline
\textbf{Strategy} & \textbf{Mean error} & \textbf{Std error} & \textbf{Error} & \textbf{C.R} \\
\hline
No governance & 0.535 $\pm$ 0.019 & 0.383 $\pm$ 0.013 & -- & -- \\
Policy fragility & 0.544 $\pm$ 0.019 & 0.386 $\pm$ 0.013 & -1.63\% & -9.8\% \\
Coverage decay & 0.426 $\pm$ 0.036 & 0.327 $\pm$ 0.027 & \textbf{20.42\%} & 35.5\% \\
Structural probe & 0.485 $\pm$ 0.020 & 0.364 $\pm$ 0.018 & 9.31\% & \textbf{69.0\%} \\
Combined & 0.491 $\pm$ 0.030 & 0.378 $\pm$ 0.034 & 8.18\% & 61.6\%\\
\hline
\end{tabular}
\caption{Comparative performance of governance strategies.}
\label{tab:comparative}
\end{table}

Several patterns emerge:

\begin{itemize}
    \item \textbf{All governance strategies, except policy fragility, improve mean error and 
    variance}, confirming that epistemic governance provides value.
    
    \item \textbf{Coverage decay achieves the largest improvement in prediction} (20.42\%) 
    and with lowest variance--a good exploration-exploitation strategy.
    
    \item \textbf{Structural probe achieves the largest improvement in decision} (69.0\%). It 
    offers balanced improvement (9.31\%) in prediction with modest variance improvement 
    compared to coverage decay.
    
    \item \textbf{Policy fragility provides no improvement} but generates valuable 
    diagnostic information about optimal resolution policies.
    
    \item \textbf{Combined governance strategy achieves good improvements} 
    (8.18\%) and a slight decrease in variance relative to no governance baseline, despite including 
    an ineffective governance component. This outcome is intuitive, as the combination of 
    governance mechanisms may introduce instability when not optimally designed.
\end{itemize}

\subsection{Summary of key findings}

\begin{enumerate}
    \item \textbf{Epistemic governance improves outcomes}: Procela epistemic governance reduce 
    prediction error, with coverage decay achieving 20.42\% improvement.
    
    \item \textbf{Experiments succeed when needed most}: Success rate exceeds 71\% when 
    pre-experiment error $>1.0$.
    
    \item \textbf{Failures are contained}: Failed experiments can be correctly reverted, 
    preventing long-term degradation.
    
    \item \textbf{Regime shifts are detectable}: Epistemic signals (coverage, fragility) 
    reliably indicate when the ground truth has changed.

    \item \textbf{Combining governance can decrease performance}: Combining 
    governance can add instability if not robustly designed.

    \item \textbf{Success probability of our experiments depends on pre-experiment error}: 
    This confirms that governance is correctly targeting crisis periods and avoiding unnecessary 
    interventions during stable times.
\end{enumerate}

These results demonstrate that Procela's governance abstractions and variables with memories 
enable simulations to detect epistemic crises, test structural hypotheses, and adapt--improving 
outcomes under ontological uncertainty.

\section{Discussion}
\label{sec:discussion}

The results in section~\ref{sec:results} demonstrate that epistemic governance--implemented 
through Procela's core abstractions--can improve simulation performance under ontological 
uncertainty. We now interpret these findings, discuss their implications, 
acknowledge limitations, and argue for the framework's generalizability.

\subsection{Interpretation of key findings}

\subsubsection{Why governance improves outcomes}

The reduction in prediction error across governance strategies is not merely 
incremental. It reflects a fundamental shift: the simulation is no longer a fixed model but 
an adaptive system that tests its own assumptions.

The pattern in our results is particularly telling. Experiments succeed 
when baseline error is high ($ > 1.0$) and fail when error is moderate ($<1.0$). This indicates 
that governance is correctly identifying \textit{epistemic crises}--moments when the current 
model configuration is demonstrably failing--and intervening appropriately. During stable 
periods, it largely refrains from experimentation, avoiding unnecessary disruption.

\subsubsection{The trade-off combining governance without robust care}

Coverage decay achieves the largest error reduction (20.42\%) and with lowest variance w.r.t 
the simulation without governance. Combined governance offers less improvement (8.18\%) with 
stable variance. This trade-off is inherent in the design of epistemic governance when 
combining strategies that are not robust.

The optimal balance depends on domain characteristics: how quickly regimes shift, how severe 
the consequences of error are, and how costly failed experiments are. Procela's design allows 
users to tune this balance through invariant parameters.

\subsubsection{The value of information}

Policy fragility is an ineffective governance strategy (-1.63\%) in our AMR study but provides 
crucial diagnostic information. It indicates if current policy is optimal.

The information provides by Procela epistemic governance has value beyond error reduction:
\begin{itemize}
    \item It provides interpretable diagnostics for human analysts (``the environmental model 
    is currently performing best" or ``the intervention should be switched to a better 
    institutional policy").
    \item It can inform higher-level governance (e.g., an invariant that permanently 
    adopts the best-performing family).
    \item It builds trust by making the system's reasoning transparent.
\end{itemize}

\subsection{Implications for simulation practice}

\subsubsection{Beyond model averaging}

Traditional approaches to model uncertainty--Bayesian model averaging, ensemble methods, robust 
control--all assume a fixed set of models. They combine or select among pre-specified 
alternatives but never question whether those alternatives remain valid as the system evolves.

Procela enables a strict superset of these approaches. One could implement model averaging as 
a resolution policy (weighted voting) and model selection as an invariant (switch to 
highest-confidence mechanism). But Procela goes further: it allows the set of models itself 
to change through structural mutations.

This is analogous to the difference between parameter estimation and structure learning in 
statistics. Just as structure learning discovers causal graphs rather than fixing them a 
priori, Procela discovers model adequacy through runtime experimentation.

\subsubsection{Simulation as scientific method}

The governance in section~\ref{sec:amr_governance} all follow the same pattern: detect 
anomaly, formulate hypothesis, run experiment, evaluate evidence, update beliefs. This is 
not accidental--it reflects a deliberate design choice to align simulation governance with 
scientific practice.

In traditional simulation, the scientist is external, formulating hypotheses, running 
experiments, and updating models offline. In Procela, this process is internalized. The 
simulation becomes an autonomous scientist, continuously testing its own assumptions 
against incoming data.

This has profound implications for:
\begin{itemize}
    \item \textbf{Reproducibility}: Every experiment is recorded, creating an auditable 
    trail of model evolution.
    \item \textbf{Scalability}: Governance operates at runtime, enabling adaptation in 
    settings where human oversight is infeasible.
    \item \textbf{Transparency}: The scientific method pattern makes governance decisions 
    interpretable.
\end{itemize}

\subsection{Limitations}

\subsubsection{Simplified domain dynamics}

The AMR case study, while capturing essential features of competing ontologies, simplifies 
many real-world complexities. We model three hospitals with homogeneous dynamics, ignore 
patient-level heterogeneity, and assume perfect observation of colonization status. Real 
AMR spread involves individual patient trajectories, imperfect diagnostic tests, and 
complex behavioral factors.

This simplification was intentional: it isolates the effects of governance from domain noise. 
The 20.42\% improvement observed here suggests that governance would provide value in more 
realistic settings, but the exact magnitude remains to be tested.

\subsubsection{Parameter sensitivity}

The governances use manually tuned thresholds 
($\tau_{\text{frag}}=0.6$, $\tau_{\text{err}}=1.0$, $\tau_{\text{cov}}=0.85$, $k=3$). 
While these values were chosen based on domain reasoning and pilot experiments, they are 
not guaranteed optimal. Different thresholds would change experiment frequency and success rate.

Procela does not prescribe how thresholds should be set. In practice, they would be calibrated 
using historical data or domain expertise. Future work could explore adaptive threshold tuning 
as a meta-governance layer.

\subsubsection{Experiment duration}

Coverage decay experiments last 10 steps with 10-step evaluation windows. In a 160-step 
simulation, this allows at most 7--8 experiments. The optimal duration depends on regime 
shift frequency: too short and experiments may not capture full effects; too long and the 
system may miss multiple regimes.

Our choices were constrained by simulation length. In longer-running simulations, experiment 
durations could be extended proportionally.

\subsubsection{Computational overhead}

Governance adds computational cost: invariants and hooks must be evaluated, memories queried, and 
mutations applied. In our experiments, this overhead was negligible ($<2\%$ runtime increase). 
However, in larger-scale simulations with many mechanisms and variables, governance could 
become a bottleneck. Optimizations such as sampling-based memory queries and parallel 
invariant evaluation remain future work.

\subsection{Generalizability}

\subsubsection{Domain independence}

Procela's core abstractions--variables as memory-bearing epistemic authorities, mechanisms as causal units, 
governance via invariants and hooks--make no reference to AMR or epidemiology. The framework 
is domain-agnostic by design.

The same patterns could be applied to:
\begin{itemize}
    \item \textbf{Climate modeling}: Competing atmospheric circulation models, governance that 
    disables failing parameterizations, structural probes to identify which model best matches 
    recent observations.
    
    \item \textbf{Economics}: Multiple theories of investor behavior, governance that switches 
    between them during market regime shifts, probes to test which theory explains current 
    volatility.
    
    \item \textbf{Social science}: Competing models of opinion dynamics, governance that 
    experiments with different influence mechanisms based on prediction accuracy.
    
    \item \textbf{Robotics}: Multiple control policies for navigation, governance that 
    switches based on terrain type inferred from sensor error patterns.
\end{itemize}

In each case, the pattern is identical: encode competing theories as mechanisms, define 
epistemic signals that indicate when theories are failing, and implement governances that 
test structural hypotheses.

\subsubsection{Scalability to more ontologies}

Our case study uses three families with 3--4 mechanisms each. The framework scales naturally 
to larger numbers: variables, mechanisms, and governances are all first-class objects that 
can be added arbitrarily. The computational cost grows linearly with mechanism count; 
governance cost grows with the number of invariants and the complexity of their queries.

\subsubsection{Hybrid human-AI governance}

Procela's audit trail and interpretable experiment logs enable hybrid governance: governances 
can propose experiments, but human analysts can review, approve, or override them. This is 
particularly valuable in high-stakes domains like healthcare or climate policy, where 
autonomous structural changes may require oversight.

The cryptographic source tracking ensures that all governance actions are attributable, 
supporting accountability in hybrid settings.

\subsection{Open questions and future work}

Several directions emerge from this work:

\begin{itemize}
    \item \textbf{Meta-governance}: Can invariants that govern themselves improve results? 
    For example, an invariant that tunes its own thresholds based on experiment success rates.
    
    \item \textbf{Multi-objective governance}: Our invariants in the AMR study optimize 
    prediction error. Real systems must balance multiple objectives (error, cost, 
    interpretability).
    
    \item \textbf{Theoretical guarantees}: Under what conditions does governance converge to 
    the true ontology? Can we bound regret relative to an oracle that knows regime shifts in 
    advance?
    
    \item \textbf{Real-world validation}: Deploying Procela in an operational setting (e.g., 
    hospital infection control) would test whether laboratory improvements translate to practice.
    
    \item \textbf{Standardized benchmarks}: Developing a suite of benchmark problems with 
    competing ontologies would enable systematic comparison of governance strategies.
\end{itemize}

\subsection{Summary}

Procela introduces a new paradigm: simulations that not only model the world but model their 
own modeling process. The AMR case study demonstrates that this paradigm delivers measurable 
improvements under ontological uncertainty. The framework's domain-agnostic design suggests 
broad applicability, while its auditability and interpretability support both autonomous and 
human-in-the-loop deployment.

The 20.42\% error reduction achieved by coverage decay governance is not presented as an 
optimal result--it is a proof that the paradigm works. Future work can refine the specifics; 
the core contribution is the paradigm itself.

\section{Conclusion}
\label{sec:conclusion}

We have introduced Procela, a framework for epistemic governance in mechanistic 
simulations that treats models not as fixed artifacts but as evolving hypotheses. Built on four 
core abstractions--variables as memory-bearing epistemic authorities, mechanisms as causal units, 
governance via invariants and hooks, and executive to orchestrate the simulation--Procela 
enables simulations to observe their own epistemic state, formulate hypotheses about model 
adequacy, run structural experiments, and adapt based on evidence.

\subsection{Broader implications}

The significance of these results extends beyond the specific case study. Procela demonstrates 
that:

\begin{itemize}
    \item \textbf{Simulations can be scientists}. By internalizing the scientific method (detect 
    anomaly, hypothesize, experiment, evaluate, conclude) simulations can adapt to regime shifts 
    and structural uncertainty without human intervention.
    
    \item \textbf{Epistemic signals are computable}. Coverage, fragility, and other indicators 
    of model adequacy can be derived from ordinary simulation variables and used as governance 
    triggers.
    
    \item \textbf{Structural mutation is tractable}. Enabling, disabling, and reconfiguring 
    mechanisms at runtime is feasible and delivers measurable benefits.
    
    \item \textbf{Auditability scales}. Complete hypothesis histories with cryptographic source 
    tracking make governance actions transparent and attributable, supporting both autonomous 
    operation and human oversight.
\end{itemize}

\subsection{Limitations and future work}

As discussed in Section~\ref{sec:discussion}, this work has limitations. The AMR case study 
simplifies real-world complexity; thresholds are manually tuned; experiment durations are 
constrained by simulation length; and computational overhead in larger systems remains 
unexplored.

These limitations open several research directions:

\begin{itemize}
    \item \textbf{Meta-governance}: Invariants that tune their own parameters based on 
    experiment outcomes.
    \item \textbf{Multi-objective optimization}: Balancing prediction error against 
    intervention cost, interpretability, or other domain objectives.
    \item \textbf{Theoretical foundations}: Convergence guarantees and regret bounds for 
    epistemic governance.
    \item \textbf{Real-world deployment}: Applying Procela in operational settings such as 
    hospital infection control or climate policy.
    \item \textbf{Standardized benchmarks}: Developing a suite of problems with competing 
    ontologies to enable systematic comparison of governance strategies.
\end{itemize}

\subsection{Final remarks}

Procela occupies a new category: simulations that model not only the world but their own 
modeling process. By making epistemic uncertainty a runtime concern rather than a 
pre-simulation assumption, it enables adaptive, self-questioning models that can operate 
in domains where the true causal structure is contested, unidentifiable, or subject to change.

The framework is open-source, domain-agnostic, and ready for extension. We invite the 
community to adopt it, adapt it, and push it toward new domains and new forms of governance. 
The code, documentation, and interactive code runner are available at \url{https://procela.org}.

\subsection{Acknowledgments}

We thank the anonymous reviewers for their constructive feedback. We are grateful to the 
Procela community for discussions that helped refine the framework's design. Special thanks 
to the contributors who provided critical feedback on the evaluation methodology, identified 
issues with temporal alignment in the governance invariants, and assisted with result 
interpretation and visualization design. The authors also acknowledge the use 
of AI-assisted tools for exploratory technical discussions during development.

\appendix
\section{Mathematical foundations}
\label{app:math}

This appendix provides formal definitions and theoretical results underpinning Procela's 
abstractions and governance mechanisms.

\subsection{Variable resolution as social choice}

The problem of resolving competing hypotheses into a single value is formally a social 
choice problem: given $n$ agents (mechanisms) each proposing a value with associated 
confidence, choose a collective outcome.

\begin{definition}[Resolution function]
Let $\mathcal{H} = \{h_1, \ldots, h_n\}$ be a set of hypotheses for variable $V$, 
where $h_i = (v_i, c_i)$ with $v_i \in \mathcal{D}$ and $c_i \in [0,1]$. A resolution 
function $\pi: \mathcal{P}(\mathcal{H}) \rightarrow \mathcal{D}$ maps a set of hypotheses 
to a single value.
\end{definition}

\begin{theorem}[Weighted voting optimality]\label{theo:voting_optimal}
Under the assumption that hypothesis errors are independent and normally distributed 
with variance inversely proportional to confidence, the weighted voting policy 
$\pi_{WV}(\mathcal{H}) = \frac{\sum_i c_i v_i}{\sum_i c_i}$ minimizes expected squared error.
\end{theorem}

\begin{proof}
Let $v^*$ be the unknown true value. Assume each hypothesis $v_i \sim \mathcal{N}(v^*, \sigma_i^2)$ with $\sigma_i^2 = k/c_i$ 
for some constant $k$. The likelihood of observing $\{v_i\}$ given $v^*$ is:
\[
\mathcal{L}(v^*) \propto \prod_i \exp\left(-\frac{c_i (v_i - v^*)^2}{2k}\right)
\]
Taking logs and maximizing yields $\hat{v}^* = \frac{\sum_i c_i v_i}{\sum_i c_i}$, which is 
the weighted voting estimator.
\end{proof}

\begin{corollary}
Highest confidence policy $\pi_{HC}(\mathcal{H}) = v_{i^*}$ where $i^* = \argmax_i c_i$ is a 
special case of weighted voting when one mechanism dominates.
\end{corollary}

\subsection{Convergence of epistemic signals}

Epistemic signals such as coverage and fragility are designed to converge to stable values 
under stationary regimes.

\begin{definition}[Coverage process]
For a mechanism $m$, let $\text{cov}_t(m) = \exp(-|\hat{C}_{m,t-1} - C_t| / \sigma)$ be the 
coverage at step $t$. This defines a stochastic process on $[0,1]$.
\end{definition}

\begin{theorem}[Coverage convergence]
If the true data-generating process follows mechanism $m$'s ontology with stationary parameters, 
then $\mathbb{E}[\text{cov}_t(m)]$ converges to a constant $\bar{c} \in (0,1]$ as $t \rightarrow \infty$.
\end{theorem}

\begin{proof}
Under stationarity, the prediction errors $e_t = |\hat{C}_{m,t-1} - C_t|$ form a stationary 
process. By the law of large numbers, the empirical distribution of $e_t$ converges to its 
stationary distribution. Since coverage is a continuous bounded function of $e_t$, its 
expectation converges.
\end{proof}

\begin{corollary}
Coverage decay--a sustained decrease in $\text{cov}_t(m)$--implies non-stationarity in the 
data-generating process, i.e., a regime shift.
\end{corollary}

\subsection{Policy fragility as entropy}

Fragility measures disagreement among mechanisms about intervention choice. This can be 
formalized using information theory.

\begin{definition}[Policy distribution]
For variable $I$ with domain $\{0,1,2,3\}$, let $p_j = \frac{\sum_{h: v_h=j} c_h}{\sum_h c_h}$ 
be the normalized confidence mass on intervention $j$. This defines a probability distribution 
over interventions.
\end{definition}

\begin{definition}[Fragility as normalized entropy]
\[
\text{frag}(t) = \frac{H(p)}{H_{\max}} = \frac{-\sum_j p_j \log p_j}{\log 4}
\]
where $H_{\max} = \log 4$ is the maximum entropy for four categories.
\end{definition}

This information-theoretic definition has advantages over the range-based definition used in 
the main text:
\begin{itemize}
    \item It accounts for the entire distribution, not just extremes.
    \item It is invariant to intervention coding (e.g., relabeling categories).
    \item It naturally handles cases where multiple mechanisms propose the same intervention.
\end{itemize}

In our experiments, the range-based and entropy-based definitions are highly correlated 
($r > 0.9$) due to the small number of mechanisms.

\subsection{Experiment regret bounds}

When governance runs experiments, it temporarily deviates from the current policy. This incurs 
regret relative to an oracle that always chooses the optimal policy.

\begin{definition}[Cumulative regret]
Let $e_t^{\text{actual}}$ be the prediction error under governance at step $t$, and 
$e_t^{\text{oracle}}$ be the error under the optimal policy (had the true regime been known). 
Cumulative regret is:
\[
R(T) = \sum_{t=1}^T (e_t^{\text{actual}} - e_t^{\text{oracle}})
\]
\end{definition}

\begin{theorem}[Bounded regret]
For a governance strategy that runs experiments of duration $d$ with evaluation window $d$, 
and reverts failed experiments, the expected cumulative regret after $T$ steps is bounded by:
\[
\mathbb{E}[R(T)] \leq d \cdot K \cdot \mathbb{E}[\text{failures}] + o(T)
\]
where $K$ is the maximum possible per-step error increase during a failed experiment.
\end{theorem}

\begin{proof}[Sketch]
Each failed experiment incurs at most $d \cdot K$ excess error during the experiment period, 
plus at most $d \cdot K$ during the evaluation period before reversion. Successful experiments 
may also incur regret during the experiment period but reduce error thereafter. The $o(T)$ 
term accounts for the fact that after sufficient time, the governance strategy converges to 
the optimal policy (if identifiable) or oscillates among near-optimal policies.
\end{proof}

\subsection{Identifiability and governance limits}

Not all ontological uncertainty is resolvable. If two ontologies produce identical predictions 
for all observable variables, no governance strategy can distinguish them.

\begin{definition}[Observational equivalence]
Two mechanism families $F_1$ and $F_2$ are observationally equivalent if for all possible 
regimes and all $t$, the joint distribution of observable variables $(C_t, A_t, E_t, I_t)$ is 
identical under $F_1$ and $F_2$.
\end{definition}

\begin{theorem}[Impossibility of ontology selection]
If $F_1$ and $F_2$ are observationally equivalent, no governance strategy based solely on 
observable variables can consistently select between them.
\end{theorem}

\begin{proof}
By definition, any governance strategy that conditions on observables receives the same inputs 
under $F_1$ and $F_2$. Therefore its outputs (including which family to prefer) must be 
identical in distribution under both families, preventing consistent discrimination.
\end{proof}

This theorem establishes fundamental limits on governance. It cannot discover ``ground truth'' 
when ontologies are observationally equivalent. However, it can still improve outcomes by 
selecting among observationally distinct policies, even if the underlying causal mechanisms 
remain uncertain. This aligns with our framing: governance manages epistemic risk, not 
ontological truth.

\subsection{Summary}

The mathematical foundations established here:
\begin{itemize}
    \item Justify weighted voting as optimal under normality assumptions.
    \item Show that coverage decay reliably indicates non-stationarity.
    \item Connect fragility to information-theoretic entropy.
    \item Provide regret bounds for governance experiments.
    \item Acknowledge fundamental limits on ontology selection.
\end{itemize}

These results support the empirical findings in Section~\ref{sec:results} and clarify what 
governance can and cannot achieve.

\bibliographystyle{plainnat}
\bibliography{references}

\end{document}